\documentclass[reprint,amsmath,amssymb,aps]{revtex4-2}

\usepackage{amscd, amsthm, graphics, mathrsfs, setspace, fancyhdr, tabularx, shapepar, xcolor, tikz, booktabs, dsfont, amsfonts, marvosym, amsbsy, bm, array,siunitx}
\usetikzlibrary{matrix,arrows,backgrounds,cd,shapes,calc}

\usetikzlibrary[topaths]
\usetikzlibrary{patterns}

\usetikzlibrary{angles,quotes}

\makeatletter
\tikzset{
        hatch distance/.store in=\hatchdistance,
        hatch distance=5pt,
        hatch thickness/.store in=\hatchthickness,
        hatch thickness=5pt
        }
\pgfdeclarepatternformonly[\hatchdistance,\hatchthickness]{north east hatch}
    {\pgfqpoint{-1pt}{-1pt}}
    {\pgfqpoint{\hatchdistance}{\hatchdistance}}
    {\pgfpoint{\hatchdistance-1pt}{\hatchdistance-1pt}}%
    {
        \pgfsetcolor{\tikz@pattern@color}
        \pgfsetlinewidth{\hatchthickness}
        \pgfpathmoveto{\pgfqpoint{0pt}{0pt}}
        \pgfpathlineto{\pgfqpoint{\hatchdistance}{\hatchdistance}}
        \pgfusepath{stroke}
    }
\makeatother

\usepackage[all]{xy}
\usepackage{mathbbol, bbm}
 \usepackage{youngtab}
 \usepackage{young}
 \usepackage{ytableau}
 \usepackage{pict2e}
 \usepackage{lipsum}
 \usepackage{float}
\usepackage{mathbbol}

\newcommand{\op}[1]{\operatorname{#1}}
\newcommand{\inv}{^{-1}}

\newcount\mycount

\newcommand{\C}{\mathbb{C}}
\newcommand{\lan}{\langle}
\newcommand{\ran}{\rangle}

\newcommand{\nin}{\not\in}

\newtheorem{innercustomthm}{Question}
\newenvironment{quest}[1]
  {\renewcommand\theinnercustomthm{#1}\innercustomthm}
  {\endinnercustomthm}

\theoremstyle{remark}
\newtheorem{thm}{Theorem}[section]
\newtheorem{lem}[thm]{Lemma}

\newtheorem{cor}[thm]{Corollary}

\theoremstyle{remark}
\newtheorem{ex}[thm]{Example}

\newcommand{\DW}[1]{\textcolor{purple}{(DW: #1)}}
 
\newcommand{\comment}[1]{}

\begin{document}
\title{Graph-Theoretic Approach to Quantum Error Correction}
\thanks{This work was presented as a poster at The Workshop on Algebraic Graph Theory and Quantum Information at the Fields Institute, and at Frontiers in Quantum Computing at University of Rhode Island. }

\author{Robert R. Vandermolen}
\email{Email: {\bf robert.vandermolen@smwc.edu}}
\affiliation{
Department of Science and Mathematics, 
Saint Mary-of-the-Woods College, 
Saint Mary-of-the-Woods, IN, USA
}

\author{Duncan Wright}
\email{Email: {\bf dwright2@wpi.edu} \\ 
Webpage: \url{https://www.wpi.edu/people/faculty/dwright2 } }
\affiliation{
Department of Mathematical Sciences, 
Worcester Polytechnic Institute, 
Worcester, MA, USA
}


%



\begin{abstract}
We investigate a novel class of quantum error correcting codes to correct errors on both qubits and higher-state quantum systems represented as qudits.
These codes arise from an original graph-theoretic representation of sets of quantum errors. In this new framework, we represent the algebraic conditions for error correction in terms of edge avoidance between graphs providing a visual representation of the interplay between errors and error correcting codes. 
Most importantly, this framework supports the development of quantum codes that correct against a predetermined set of errors, in contrast to current methods. A heuristic algorithm is presented, providing steps to develop codes that correct against an arbitrary noisy channel. 
We benchmark the correction capability of reflexive stabilizer codes for the case of single qubit errors by comparison to existing stabilizer codes that are widely used. In addition, we present two instances of optimal encodings: an optimal encoding for fully correlated noise which achieves a higher encoding rate than previously known, and a minimal encoding for single qudit errors on a four-state system.
\end{abstract}


\maketitle



\section{Introduction}
Error correcting codes are essential tools in communication theory as they provide the means for the reliable delivery of data over noisy communication channels. In classical computing theory we have found that the ability to correct single-bit errors is not only fundamental, but sufficient for most purposes \cite{eccfundamentals}.
This classical computing mindset has influenced the current approach to quantum error correction, with the majority of work focusing on the correction of single-qubit flip, phase and phase-flip errors, characterized by tensors of Pauli-spin operators \cite{ste98,errortheory,nielsen2002quantum,ortho,nonbinary}.

Amongst the approaches influenced by classical computing, surface codes have made remarkable achievements, for instance, proving that fault tolerance is theoretically possible to achieve once certain levels of fidelity are reached \cite{surf1,surf2,surf3}.
However, the proof relies on the assumption that the errors to be corrected are uncorrelated across both time and space. 
Recently, the validity of this assumption has been brought into question, by the experimental observation of correlated errors across both time and space \cite{correlated}. Practically, the assumption of uncorrelated errors can lead to lower rates of error correction and fidelity \cite{temporalcorrelation}. For these reasons, the ability to correct correlated errors has become increasingly relevant. 
Moreover, the wealth of research in engineering a quantum computer has resulted in a wide variety of architectures such as superconducting qubits\cite{superconducting}, quantum dots\cite{dots}, trapped ions\cite{trappedions}, photonics\cite{photons}, and more \cite{quantumcomputing}.
Each systems' qubit architecture comes with an intrinsic error set, incentivizing error correcting codes that are developed to correct a pre-defined set of errors. 
 
The reflexive stabilizer codes introduced in this manuscript are not only capable of correcting correlated errors, but the framework allows for the development of codes that correct any given error set. 
This novel approach to quantum error correction uses edge avoidance in a special class of graphs to avoid arbitrary error sets, including those correlated across space.  Furthermore, these codes are developed for qudits, allowing for implementation when more than two energy levels are measurable, such as the silicon-based quantum dot \cite{quditsnotqubits}. 
All codings will be done into strings of qudits, represented by $\C^d$, the computational basis of dimension $d$. It is worth noting that this is not the first application of graph theory in quantum error correcting codes, see e.g.\ \cite{graphstates}.

The paper is organized as follows. In Section~\ref{old error} we recall essential background in quantum error correction, stabilizer codes, and graph theory.  Definitions and useful properties are presented for, amongst others, the Pauli error operators and Cayley graphs. In Section~\ref{cayley sect} we give the specific graphs we consider to encode quantum errors. Next, in Section~\ref{reflexive} the novel reflexive stabilizer codes are introduced alongside their graph representation. To show the initial benefit of of these new codes, we give a novel minimal encoding of a single qudit on a 4-state system in Subsection~\ref{perfectcoding}. Further, in Subsection~\ref{fullycorrelated} we consider fully correlated noise and achieve an optimal encoding by reflexive stabilizer codes, improving on the results of \cite{full}. 
We present a heuristic algorithm to build a reflexive stabilizer code that corrects a given error set in Section~\ref{heu} before concluding in Section~\ref{discussion}.


\section{Preliminaries}\label{old error}

We briefly review the relevant terms and notations used throughout this manuscript. 
All codewords will be represented as strings (or the superposition of strings) of qudits from the quantum $d$-ary alphabet $\C^d$, where $d=p^m$ such that $p$ is prime and $m$ is an integer. Unlike classical computing, we also consider codewords that are the superposition of those strings from the computational basis. Moreover, we set $\omega=\text{exp}(2\pi i/p)$ as the primitive $p^\text{th}$ root of unity. 

As with other stabilizer codes (see e.g.\ \cite{goodcodes, nonbinary}), errors will be labeled with strings from the field $\mathbb{F}_d$ on $d=p^m$ elements. Given a linear subset $C\subseteq \mathbb{F}_d^n$, we denote by $C^\perp$ the orthogonal subspace with respect to the inner product $\left\lan\,a,b\,\right\ran=\sum_{i=1}^na_i b_i$, for a chosen basis $\{e_i\}$ of $\mathbb{F}_d^n$ over $\mathbb{F}_d$.
That is 
\begin{equation}\label{Cperp}
C^\perp=\left.\{ v\in\mathbb{F}_d^n\,\right|\,\lan v,a\ran=0,\,\,\,\forall\,a\in C\}
\end{equation}
We define the weight of $C$ as the minimum Hamming weight $\op{w}(c)$; i.e.\ the number of non-zero entries, across all the elements $c\in C$, 
\begin{equation}
\op{wt}(C)=\op{min}\left.\Big\{\op{w}(c)\,\right|\,c\in C\setminus\{\overline{0}\}\Big\},
\end{equation}
where $\overline{0}$ is the string of all zeroes. We similarly denote the string of all ones by $\overline{1}$.

For ease of calculations, we fix a basis for $\mathbb{F}_d$ over $\mathbb{F}_p$, labeled $\{f_i\,:\, i\in\{1,\ldots,m\}\}$, and represent elements in terms of this basis as, for example, $a=\sum_{i=1}^m\alpha_if_i$ and $b=\sum_{i=1}^m\beta_if_i$. Furthermore, given these representations, we define the inner product $*$ on $\mathbb{F}_d$ by 
\begin{equation}\label{inner_field}
    a*b=\sum_{i=1}^m\alpha_i\beta_i.
\end{equation}


\subsection{The error group}
Single qudit errors are defined using the generalized Pauli matrices, $X(a)$ and $Z(b)$ for each $a,b\in\mathbb{F}_d$, whose action on $|x\ran\in\C^d$ is given by 
\begin{equation}\label{paulis}
 X(a)|x\ran=|x+a\ran \;\;\text{ and }\;\; Z(b)|x\ran=\omega^{b*x}|x\ran,
\end{equation}
where $\omega$ is the primitive $p^\text{th}$ root of unity. The operators $X(a)$ and $Z(b)$ are referred to as the flip and phase errors, respectively. We will refer to the operator $Y(a)=\omega X(a)Z(a)$ as the phase-flip error.  
For qubits; i.e.\ when $p=d=2$, one quickly notes that the standard Pauli matrices are given by 
\begin{subequations}\label{qubit basis}
\begin{align}
X(0)=Z(0)=\mathbb{1}_2\;\;\;\;X=X(1)=\begin{bmatrix}0&1\\1&0\end{bmatrix}\\
Z=Z(1)=\begin{bmatrix}1&0\\0&-1\end{bmatrix}\;\;\;\;Y=iXZ=\begin{bmatrix}0&-i\\i&0\end{bmatrix} 
\end{align}
\end{subequations}


Notice that the inner product $*$, appearing in Equation~\eqref{paulis} and defined in Equation~\eqref{inner_field}, is in one-to-one correspondence with a trace operator $tr_*: \mathbb{F}_d\mapsto\mathbb{F}_p$ defined by the basis $\{f_i\}_i^m$. Often in the literature (see e.g.\ \cite{nonbinary}), a trace operator is used in the definition of the generalized Pauli operators, yet the definition is independent of choice of trace operator (see again \cite{nonbinary}).
For this reason, in this manuscript, we make use of the inner product definition, choosing a notation similar to that in \cite{ortho}. 

For errors on an $n$-qudit system, we concatenate the Pauli operators to define the \textbf{error operator} 
\begin{equation}
D_{a,b}=X(a_1)Z(b_1)\otimes...\otimes X(a_n)Z(b_n)
\end{equation}
for each $a=(a_{1},\ldots,a_{n}),b=(b_{1},\ldots,b_{n})\in\mathbb{F}_d^{n}$. 
One will verify that
\begin{equation}\label{Dab props}
D_{a,b}\inv = D_{-a,-b}
\quad\text{and}\quad 
D_{a,b}D_{ c, d} = \omega^{-\left\lan b, c\right\ran} D_{a+ c,b+ d}.
\end{equation}
Hence, the collection of $n$-qudit errors generates the multiplicative \textbf{error group} 
\begin{equation}
\mathcal{E}_n=\left\{\omega^\kappa D_{a,b}\;\;\Big|\;\; a,b\in\mathbb{F}_d^{n},\,\, \kappa\in\{0,\ldots,p-1\}\,\right\}.
\end{equation}
We will refer to any non-trivial subset $\mathscr{E}$ of $\mathcal{E}_n$ as an \textbf{error set}. 
We will assume that $\mathbb{1}=D_{\overline 0, \overline 0}$ is in every error set as one should always protect against no error. 

\noindent
For a more thorough introduction to quantum error correction, see e.g. \cite{fault, goodcodes, gf4, calcliff, nonbinary, ortho,errortheory,mixed,ste98}. Next, we introduce the basics of quantum stabilizer codes. 

\subsection{Stabilizer Codes}\label{stabilizer}
Briefly, the objective of quantum stabilizer codes is to be able to protect from any error of a commutative subgroup of errors $S$ and correct any error from a larger set of errors $S\subset\mathscr{E}\subset\mathcal{E}_n$. The reader less familiar with stabilizer codes is referred to \cite{ste98, ortho, goodcodes, errortheory}.

Let $S$ be a commutative subgroup of errors containing the center $\mathcal{Z}$. A \textbf{quantum stabilizer code} $R$ is any joint eigenspace of the operators in $S$. We refer to $S$ as the \textbf{stabilizer} of $R$. In practice, $R$ will be represented by a collection of orthogonal eigenvectors $|\Phi_1\ran,\ |\Phi_2\ran, \cdots$, which we refer to as code words. Necessary and sufficient conditions for $R$ to protect from any error in a given error set $\mathscr{E}$ have been established in \cite{errortheory,mixed} and are as follows: for any two distinct code words $|\Phi_1\ran$ and any two errors $E_1,E_2\in\mathscr{E}$, we must have
\begin{subequations}\label{qecc conds}
\begin{align}
\lan \Phi_1|E_1\inv E_2|\Phi_2\ran & =0,\\
\lan \Phi_1|E_1\inv E_2|\Phi_1\ran & =\lan \Phi_2|E_1\inv E_2|\Phi_2\ran.
\end{align}
\end{subequations}

Intuitively, these conditions guarantee that regardless of the errors that might occur to distinct code words, their perturbed states remain distinguishable by quantum measurement and have equal weight. 
It is of note that we will always assume $\mathbb{1}\in\mathscr{E}$ as one should always protect from no errors occurring. 

Due to its appearance in Equation~\ref{qecc conds} and its pervasiveness in the theory, we will refer to $E_1\inv E_2$ as a \textbf{conjugate error} of $\mathscr{E}$. Furthermore, we will denote the set of conjugate errors of $\mathscr{E}$ by 
\begin{equation}\label{E2}
\mathscr{E}^2= \left.\Big\{E_1\inv E_2\, \right|\  E_1,E_2\in\mathscr{E}\Big\}.
\end{equation}

For the error correction capabilities of stabilizer codes an additional subset of errors is important, namely the \textbf{centralizer} of $S$, which we denote by $S^\perp$, i.e. the elements of $\mathcal{E}_n$ which commute with all elements of $S$. 
The following theorem follows identically to that in \cite[Theorem 1]{ortho} and \cite[Theorem 3]{nonbinary}. We include it for later reference, and ease of the reader. 

\begin{thm}\label{errorcorrecting}
Let $S$ be a commutative subgroup of $\mathcal{E}_n$ which contains the center, i.e.\ $\mathcal{Z}\subset S$.
Further, let $\mathscr{E}\subset\mathcal{E}_n$ be an error set.  Then any stabilizer code for $S$ is an error-correcting code which will correct any error from $\mathscr{E}$ if and only if every conjugate error $E\in\mathscr{E}^2$ satisfies either $E\in S$ or $E\nin S^\perp$. 
\end{thm}

Commutativity of the error operators is characterized by the relation:
\begin{equation}\label{eqcomm}
D_{a,b}D_{ c, d}=\omega^{(a,b)\star( c,b)}D_{ c, d}D_{a,b}
\end{equation}
where 
\begin{equation}\label{star}
(a,b)\star( c, d)=\left\lan\, b, c\,\right\ran-\left\lan\,a, d\,\right\ran\end{equation}
is the symplectic inner product on $\mathbb{F}_d^{2n}$.
Explicitly, two error operators $\omega^\kappa D_{a,b}$ and $\omega^{\kappa'}D_{ c, d}$ commute if and only if $(a,b)\star( c, d)=0$.
\comment{\begin{equation}\label{star comm}
    (a,b)\star( c, d)=0.
\end{equation}}
One may recognize that these are simply a discrete form of the commutativity relations for the Weyl operators. 

\subsection{Graph Theory}
Graph theory has been used both theoretically and experimentally in quantum error correcting codes before. For instance, graph state codes have been studied and generated experimentally as a form of measurement-based quantum computation \cite{graphstates}. In this work, we rely on graph theory as well. 

The main novelty of the current manuscript lies in the graph-theoretic representation of quantum stabilizer codes and the associated error sets. Before introducing these specialty graphs in the following section, we must first recall some of the basic terminology of graphs. 

A graph $G=(V,E)$ consists of a collection of vertices $V$ connected by a set of edges $E$. Throughout, we consider only undirected graphs meaning simply that the edges are directionless. 
Two vertices $u$ and $v$ connected by an edge $e=(u,v)$ are said to be adjacent, and the edge $e$ is said to be incident to both $u$ and $v$. 
Whenever $u=v$, the edge $(u,u)$ is referred to as a \textbf{loop}; we will also refer to the looped vertex $u$ as a loop. Whenever $u$ and $v$ are distinct, we refer to the edge $(u,v)$ as a \textbf{simple} edge. 
A graph is 
\textbf{complete} whenever each pair of distinct vertices is adjacent. 

In what follows, it is convenient to distinguish between the set of loops and the set of simple edges of a graph $G=(V,E)$. To this end, we will decompose the edge set $E$ into two sets: The set of simple, or non-looped, edges will be denoted by $\hat E$ and will be referred to as the \textbf{simple edges of $\mathbf{G}$}, and we denote by $\mathscr{L}_G$ the \textbf{loop set of $\mathbf{G}$}. 
When convenient, and without risk of confusion, we will refer to a vertex $a\in\mathscr{L}_\mathscr{E}$ as opposed to the more accurate $(a,a)\in\mathscr{L}_\mathscr{E}$. 

The \textbf{component} of a vertex $v\in V$ is the subgraph of $G$ consisting of only those vertices $V'$ that can be reached from $v$ (via consecutive edge-traversals) and those edges incident to the vertices in $V'$. 
Lastly, the \textbf{complement} of a (simple) graph $G$ is the graph $G^\perp$ on the same vertex set such that two vertices are adjacent in $G^\perp$ if and only if they are not adjacent in $G$. 
For a more thorough introduction to graph theory see e.g.\ \cite{graphtheorybook}.

\section{Graphs for quantum error correcting}\label{qecc graphs}
Next we develop the specialty graphs that we utilize to assist in and visualize quantum error correction. The novel  graph-theoretic representations of both quantum stabilizer codes and error sets allows one to easily identify the errors that an encoding will correct and, more importantly, an encoding that will correct a pre-determined error set. Specifically, we will define an error avoidance graph for an arbitrary error set, as well as linear undirected Cayley graphs. 

\subsection{LUC Graphs}\label{cayley sect}
In this manuscript, we will use Cayley graphs to represent the encoding of our novel stabilizer codes. In short, Cayley graphs provide a means to represent a group action as a graph. 
A Cayley graph $G_C=(V_C,E_C)$ is defined by a subset $C$ of its vertex set known as the \textbf{connecting set}. We restrict our attention to Cayley graphs whose connecting set $C$ is a linear subgroup of the the additive group $V_C=\mathbb{F}_d^n$. In this case the edge set $E_C$ is the collection of pairs $(a,a+c)$ for $a\in V_C$ and $c\in C$. 
Note that, since the group identity $\overline{0}$ is necessarily in $C$, there is a loop at every vertex. 
We will refer to such a $G_C$ as a  \textbf{linear undirected Cayley (LUC) graph}. 

In what follows we will be interested in subgraphs of LUC graphs obtained by deleting particular loops. Given a linear subset $C_1\subseteq C$ we set $G_C^{C_1}=(V_C^{C_1},E_C^{C_1})$ to be the subgraph of $G_C$ with edge set $E_C^{C_1}=\hat E_C \cup \mathscr{L}_{C_1}$ where $\mathscr{L}_{C_1}=C_1^\perp$.
That is, the only loops remaining from $G_C$ are at the vertices in $C_1^\perp$. The reason for this convention will become clear in Section~\ref{reflexive}. We will also refer to any graph $G_C^{C_1}$ as a LUC graph. Notice that $G_C^{\{\overline 0\}}=G_C$. We consider such a LUC graph in the following example. 
For more on Cayley graphs, in their full generality, 
see \cite{cayley}.


\begin{ex}\label{example1}
Consider a three-state quantum system ($d=3$) of two qudits ($n=2$) and a connecting set $C=\{00,11,22\}\subseteq \mathbb{F}_3^2$. In this case, we have the LUC graph $G_C=G_C^{\{\overline 0\}}$ with vertex set $V_C=\mathbb{F}_3^2$ 
and edge set $E_C=\hat E_C \cup \mathscr{L}_{\{\overline 0\}}$ shown in  Figure~\ref{fig: ex1}.  
Note that $G_C$ consists of three complete components and that the component containing $00$ contains exactly those vertices in $C$. These observations can be generalized and are made formal in the following theorem. 
\end{ex}


\begin{figure}[h]
\begin{center}
\begin{tikzpicture}
  [scale=.6,auto=left,every node/.style={circle,fill=blue!20}]
 
 \node (n0) at  (0,5)  {00};
 \node (n1) at  (-3.21,3.83) {10};
 \node (n2) at (-4.924,0.868) {20};

 \node (n3) at (-4.33,-2.5) {01};

 \node (n4) at (-1.71,-4.698) {11};
 \node (n5) at (1.71,-4.698) {21};
 
 \node (n6) at (4.33,-2.5) {02};

\node (n7) at (4.924,0.868) {12};
\node (n8) at (3.21,3.83) {22};

  \foreach \from/\to in {n0/n4,n0/n8,n8/n4}
    \draw[ultra thick] (\from) -- (\to);
    
    \foreach \from/\to in {n1/n5,n1/n6,n5/n6}
    \draw[blue,ultra thick] (\from) -- (\to);
    
    \foreach \from/\to in {n2/n3,n2/n7,n7/n3}
    \draw[red,ultra thick] (\from) -- (\to);
    

\end{tikzpicture}
\end{center}
\caption{The LUC graph $G_C=G_C^{{\{\overline 0\}}}$ with connecting set given in Example~\ref{example1}. The connecting set $C=\{00,11,22\}\subset V_C=\mathbb{F}_3^2$ (shown in black) is one of exactly three complete components. Since $\mathscr{L}_{\{\overline 0\}}
=\mathbb{F}_3^2$, there are loops at every vertex.  }
\label{fig: ex1}
\end{figure}
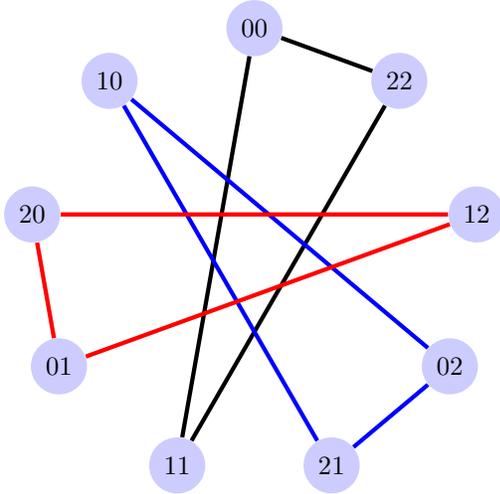

\begin{thm}\label{cayley thm}
Let $C$ be a linear subspace of $\mathbb{F}_d^n$ where $d=p^m$. The LUC graph $G_C$ has exactly $d^{n-\op{dim}(C)}$ number of complete components. 
Moreover, the component which contains $\overline{0}$ is exactly $C$.
\end{thm}
\begin{proof}
Due to linearity of $\mathbb{F}_d^n$, it is enough to show the connected component, $H$, containing $\overline{0}$ is complete and contains exactly the elements from $C$. It is clear that $\overline{0}$ is connected to exactly those $c\in C$. Moreover, if $b\notin C$, then $b-a\notin C$ for any $a\in C$. Hence the connected component containing $\overline{0}$ contains exactly the elements of $C$. Lastly, by linearity of $C$, if $a,\ b\in C$, then $b-a\in C$ and thus $(a,b)\in \hat E_C$. Thus $H$ is complete. 
\end{proof}


\subsection{The Error Avoidance Graph}\label{error graph}
Next, we show how to associate a graph to the conjugate errors of an error set. This association allows one to determine an encoding that will correct a pre-determined set of errors, in contrast to many current quantum error correcting codes. Moreover, we will show that a such an encoding can be found relying only on graph-theoretic principles.   

Fix an error set $\mathscr{E}$ acting on a $d$-state quantum system of $n$ qudits. 
We denote by  $G_\mathscr{E}=(V_\mathscr{E},E_\mathscr{E})$ the graph with vertices $V_\mathscr{E}=\mathbb{F}_d^n$ and edge set 
\begin{equation}
E_\mathscr{E}=\{(a,b)\,|\,\omega^\kappa D_{ab}\in\mathscr{E}^2\}.
\end{equation}
In essence, $E_\mathscr{E}$ encodes the conjugate errors $\mathscr{E}^2$ by effectively connecting codewords that would fail the necessary distinctness conditions given in Equation~\eqref{qecc conds}. (This idea will be made formal in the next section.) 
For this reason, we refer to $G_\mathscr{E}$ as the \textbf{error avoidance graph}. 

The non-trivial loops indicate the possible strings at which flip and phase errors must occur simultaneously in $\mathscr{E}^2$, whereas the trivial loop (at $\overline{0}$) indicates no error occurring. 


\begin{ex}\label{error graph example}
We consider a system of three qubits with correlated errors from the set 
\begin{equation}
\mathscr{E}=\{\mathbb{1},D_{e_1,e_2},D_{e_3,e_3},D_{e_2,e_1}\}. 
\end{equation}
The non-loop edges are given by 
\begin{equation}
\hat E_\mathscr{E}=\{(010,100),(011,101)\}
\end{equation}
and the loop set is given by 
\begin{equation}
\mathscr{L}_\mathscr{E}=\{000,110,001\}. 
\end{equation}
The graph $G_\mathscr{E}$ is shown in Figure~\ref{fig: dw1} with the loop set indicated by dark blue nodes.
\end{ex}



 


 


    


\section{Reflexive Stabilizer Codes}\label{reflexive}
In this section, we define a novel class of quantum stabilizer codes which arise from the connecting sets of LUC graphs. 
It is the interplay between the error avoidance and LUC graphs which allows us to develop our new class stabilizer codes with error avoidance at the forefront. 
 


Let $C$ be a linear subspace of $\mathbb{F}_d^n$ and $C_1\subset C$. The \textbf{reflexive stabilizer} of $C$ with respect to $C_1$ is the subgroup of the error group $\mathcal{E}_n$ 
generated by
\begin{equation}\label{refl stab}
S_C^{C_1}=\left\lan D_{aa},\, D_{b0}\,:\, a\in C^\perp, \,b\in C_1 \right\ran
\end{equation}
One quickly sees that $S_{C}^{C_1}$ is commutative as $( a, a)\star( b, b)=0$ and $( a, a)\star( b,\overline{0})=0$ by Equations~\eqref{star} and \eqref{Cperp}. Therefore a reflexive stabilizer is indeed a quantum stabilizer. 
The following lemma gives the form of the centralizer of a reflexive stabilizer $S_C^{C_1}$. 

\begin{lem}\label{sperp}
Let $C$ be a linear subspace of $\mathbb{F}_d^n$ and $C_1\subset C$, then the centralizer of $S_{C}^{C_1}$ is given by
\begin{equation}
\left(S_C^{C_1}\right)^\perp=\left\langle D_{ab}\;\Big|\; a-b\in C,\;\text{and } a\in C_1^\perp \right\rangle.
\end{equation}
\end{lem}
\begin{proof}
This follows from Equations ~\eqref{eqcomm} and \eqref{star} as, for any $x\in C^\perp$, we have $(x,x)\star(a,b)=0$ exactly when
$x\cdot (a- b)=0.$
Also, for any $y\in C_1$, we have $(y,0)\star(a,b)=0$ only when $y\cdot a=0$.
\end{proof}

Recall that a quantum stabilizer code is any joint eigenspace of the operators in its stabilizer. We will denote by $R_C^{C_1}$ the \textbf{reflexive stabilizer code} (RSC) with reflexive stabilizer $S_C^{C_1}$. One quickly notes that an RSC $R^{C_1}_C$ will encode $k=\text{dim}(C)-\text{dim}(C_1)$ logical qudits into an $n$ physical qudit system. 
A constructive form of reflexive stabilizer codes is given in Appendix~\ref{app: constructive}. 

The following theorem summarizes the errors sets that $R_C^{C_1}$ can correct; it is simply a rewording of Theorem~\ref{errorcorrecting} in terms of the LUC and error avoidance graphs. The details of the proof can be found in Appendix~\ref{app: lemmas}. 

\begin{thm}\label{graph correct}
Let $d=p^m$ for some prime $p$, $C_1\subset C\subset\mathbb{F}_d^n$ be linear subspaces, and let $\mathscr{E}$ be an error set. 
If the only edges common to both $G_\mathscr{E}$ and $G_C^{C_1}$ are 
incident to $\overline{0}$ or a vertex outside $C_1^\perp$, 
then the reflexive stabilizer code $R_C^{C_1}$ can correct any error $\mathscr{E}$. In short, if we have 
\begin{equation}\label{graph correct eqn}
E_C^{C_1}\cap {E}_\mathscr{E}\subseteq\{(a,b)\,|\,a\nin C_1^\perp\;\text{or}\;a=0\}.
\end{equation}
\end{thm}


The following corollary is a simplification of Theorem~\ref{graph correct} which is easier to verify. 

\begin{cor}\label{cor: graph correct}
Let $d=p^m$, $C_1\subset C\subset\mathbb{F}_d^n$, and let $\mathscr{E}$ be as in Theorem~\ref{graph correct}. If 
\begin{equation}\label{cor eqn}
E_C^{C_1}\cap E_\mathscr{E}=\{(\overline{0},\overline{0})\}  
\end{equation}
then $R_C^{C_1}$ can correct any single error from $\mathscr{E}$. 
\end{cor}

By Corollary~\ref{cor: graph correct}, finding a reflexive stabilizer code capable of correcting an error set $\mathscr{E}$ is as simple as finding a connecting set $C$ such that $G_C$ avoids the edges of $G_\mathscr{E}$. This is illustrated in the following example. It is worth noting that this condition, as opposed to that in Theorem~\ref{graph correct}, does not always allow for an encoding of the maximum number of physical qudits.

\begin{ex}\label{dw1}
Consider the error set $\mathscr{E}=\{\mathbb{1}, D_{e_1,e_2}, D_{e_3,e_3}, D_{e_2,e_1}\}$ on the system of three qubits discussed in Example~\ref{error graph example}. 
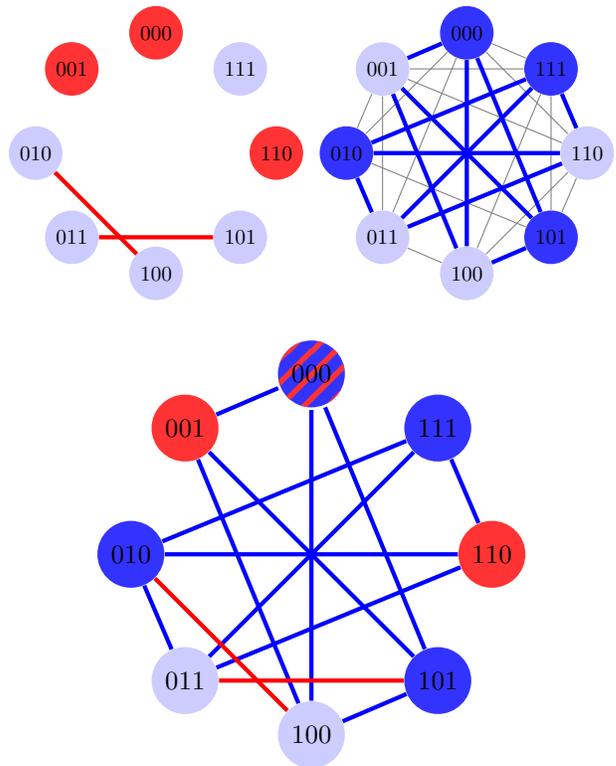
\begin{figure}
\begin{center}
 


 

 

    
    


\begin{tikzpicture}
  [scale=.4,auto=left,every node/.style={circle,scale=0.8,fill=blue!20}]
 
 \node[fill=red!80] (n0) at  (0,4)  {000};
 \node[fill=red!80] (n1) at  (-2.8,2.8) {001};
 \node (n2) at (-4,0) {010};

 \node (n3) at (-2.8,-2.8) {011};

 \node (n4) at (0,-4) {100};
 \node (n5) at (2.8,-2.8) {101};
 
 \node[fill=red!80] (n6) at (4,0) {110};

\node (n7) at (2.8,2.8) {111};

  \foreach \from/\to in {n2/n4,n3/n5}
    \draw[red,ultra thick] (\from) -- (\to);
    

\end{tikzpicture}\hspace{1mm}
\begin{tikzpicture}
  [scale=.4,auto=left,every node/.style={circle,scale=0.8,fill=blue!20}]
 
 \node[fill=blue!80] (n0) at  (0,4)  {000};
 \node (n1) at  (-2.8,2.8) {001};
 \node[fill=blue!80] (n2) at (-4,0) {010};

 \node (n3) at (-2.8,-2.8) {011};

 \node (n4) at (0,-4) {100};
 \node[fill=blue!80] (n5) at (2.8,-2.8) {101};
 
 \node (n6) at (4,0) {110};

\node[fill=blue!80] (n7) at (2.8,2.8) {111};
 
  \foreach \from/\to in {n0/n1,n0/n2,n0/n3,n0/n4,n0/n5,n0/n6,n0/n7}
    \draw[gray] (\from) -- (\to);
     \foreach \from/\to in {n1/n2,n1/n3,n1/n4,n1/n5,n1/n6,n1/n7}
    \draw[gray] (\from) -- (\to);
     \foreach \from/\to in {n2/n3,n2/n5,n2/n6,n2/n7}
    \draw[gray] (\from) -- (\to);
  \foreach \from/\to in {n3/n4,n3/n6,n3/n7}
    \draw[gray] (\from) -- (\to);
  \foreach \from/\to in {n4/n5,n4/n6,n4/n7}
    \draw[gray] (\from) -- (\to);
  \foreach \from/\to in {n5/n6,n5/n7}
    \draw[gray] (\from) -- (\to);
  \foreach \from/\to in {n6/n7}
    \draw[gray] (\from) -- (\to);

  \foreach \from/\to in {n0/n1,n0/n5,n1/n5,n0/n4,n1/n4,n4/n5}
    \draw[blue,ultra thick] (\from) -- (\to);
    
  \foreach \from/\to in {n7/n6,n7/n2,n6/n2,n7/n3,n6/n3,n2/n3}
    \draw[blue,ultra thick] (\from) -- (\to);
    

\end{tikzpicture}

\vspace{5mm}
\begin{tikzpicture}
  [scale=.6,auto=left,every node/.style={circle,fill=blue!20}]
   \node[preaction={fill=blue!80}, pattern=north east hatch, pattern color=red!80, hatch distance=10pt, hatch thickness=2pt] (n0) at (0,4) {000};
 \node[fill=red!80] (n1) at  (-2.8,2.8) {001};
 \node[fill=blue!80] (n2) at (-4,0) {010};
 \node (n3) at (-2.8,-2.8) {011};
 \node (n4) at (0,-4) {100};
 \node[fill=blue!80] (n5) at (2.8,-2.8) {101};
 \node[fill=red!80] (n6) at (4,0) {110};
\node[fill=blue!80] (n7) at (2.8,2.8) {111};
 
    
  \foreach \from/\to in {n0/n1,n0/n5,n1/n5,n0/n4,n1/n4,n4/n5}
    \draw[blue,ultra thick] (\from) -- (\to);
    
  \foreach \from/\to in {n7/n6,n7/n2,n6/n2,n7/n3,n6/n3,n2/n3}
    \draw[blue,ultra thick] (\from) -- (\to);
    
      \foreach \from/\to in {n2/n4,n3/n5}
    \draw[red,ultra thick] (\from) -- (\to);

\end{tikzpicture}
\end{center}
\vspace{-.2in}
\caption{(Top Left) The error avoidance graph $G_\mathscr{E}$ for the error set $\mathscr{E}=\{\mathbb{1}, D_{e_1,e_2}, D_{e_3,e_3}, D_{e_2,e_1}\}$. Red nodes indicate loops from  $\mathscr{L}_\mathscr{E}$. (See Example~\ref{error graph example}.) 
(Top Right) The LUC graph $G_C^{C_1}$ from Example~\ref{dw1} with $C=\{000,100,001,101\}$ and $C_1=\langle 101\rangle$ shown in blue. Gray edges indicate the unused simple edges  obtained from the complement of $G_\mathscr{E}$. 
(Bottom) The two graphs $G_\mathscr{E}$ and $G_C^{C_1}$ drawn together. Notice that the only edge in common is the loop at $\overline 0$, satisfying the condition of Corollary~\ref{cor: graph correct}. }
\label{fig: dw1}
\end{figure}
We must first choose a connecting set $C$ such that the simple edges, $\hat E_C$, of $G_C$ avoid those in $G_\mathscr{E}$. One possible option is 
\begin{equation}
C=\{000,100,001,101\}.
\end{equation}
Next, we choose a subspace $C_1\subseteq C$ so that the loop set, $\mathscr{L}_{C_1}$, of $G_C^{C_1}$ intersects $\mathscr{L}_\mathscr{E}$ only at $(\overline 0, \overline 0)$. 
Thus Equation~\eqref{cor eqn} is satisfied and we have that $R_C^{C_1}$ will correct any single error from the given error set $\mathscr{E}$. (See Figure~\ref{fig: dw1}.)  
\end{ex}


In Appendix~\ref{app: benchmark} we benchmark reflexive stabilizer codes against the well-known CSS codes. Specifically, we show that reflexive stabilizer codes have the same signle qubit error correction capabilities as their CSS counterparts.  
Moreover, we give a one-to-one correspondence between them.

\section{Heuristic algorithm}\label{heu}
In this section we will lay out the concise steps for a heuristic algorithm to build a reflexive stabilizer code which will correct a given error set according to Theorem~\ref{graph correct}. 
Briefly, the steps are as follows: (S0) Fix an error set $\mathscr{E}$. (S1) Construct its error avoidance graph $G_\mathscr{E}$. (S2) Find a $C_0$ whose LUC graph $G_{C_0}$ avoids the simple edges of $G_\mathscr{E}$. (S3) If possible, extend $C_0$ to a subspace $C$ and simultaneously choose a subcode $C_1\subset C$ that together satisfy Equation~\eqref{graph correct eqn}. 
Theorem~\ref{graph correct} then gives confirmation that the code $R_C^{C_1}$ can correct any single error from $\mathscr{E}$.

\vspace{2mm}
\noindent
\textbf{Step 0:} \textbf{Start with an error set} $\mathbf{\mathscr{E}}$

One major benefit of RSCs is that they provide a means of finding a code that correct against a pre-determined error set associated to a noisy channel. This is in contrast to most codes in the literature where one starts with a code and then searches for the errors it it corrects. 
In this manuscript, we choose error sets that are either convenient for theoretical analysis or are illustrative of the concepts we are developing. 
However, future works will focus on error sets that appear in physical quantum computers such as those found in \cite{correlated}.

\vspace{2mm}
\noindent
\textbf{Step 1:} \textbf{Build the error avoidance graph} $\mathbf{G_{\mathscr{E}}}$

The error avoidance graph $G_{\mathscr{E}}$, as defined in Section~\ref{reflexive}, encodes the conjugate errors produced by a noisy channel as edges. Besides the certain conditions outline earlier, we wish to avoid these edges with the LUC graph of a reflexive stabilizer code.  

\vspace{2mm}
\noindent
\textbf{Step 2: Find a LUC graph} $\mathbf{G_{C_0}}$ \textbf{which avoids} 
$\mathbf{\hat E_{\mathscr{E}}}$

Find a maximal connecting set $C_0$ whose LUC graph $G_{C_0}$ avoids the simple edges $\hat E_{\mathscr{E}}$ of $G_\mathscr{E}$. 
Starting with this LUC graph, one can obtain a lower bound on the rate of error correction by exploring possible subspaces of $C_0$ that satisfy Equation~\eqref{cor eqn}. However, Theorem~\ref{graph correct} allows for a weakening of this condition, thereby allowing for a higher rate of error correction. Extending $C_0$ for this purpose is addressed in the next step.
Sometimes it is enough to not extend $C_0$, see e.g.\ Examples~\ref{dw1}.

\vspace{2mm}
\noindent
\textbf{Step 3: Loop avoidance and extensions}

We attempt to find a linear extension $C\supset C_0$ and a linear subspace $C_1\subset C$ such that the conditions of Theorem~\ref{graph correct} are satisfied. In particular, we need $C_1^\perp\cap\mathscr{L}_\mathscr{E}=\{\overline 0\}$ and $\hat E_C\cap \hat E_\mathscr{E}\subset\{(a,b)\,|\,a\nin C_1^\perp\}$.

\vspace{2mm}
\noindent
\textbf{Conclusion}

Finally, the reflexive stabilizer code $R_C^{C_1}$ can be defined using the $C_1$ and $C$ from Step 3. 
Following Theorem~\ref{graph correct}, $R_C^{C_1}$ will correct any single error from the error set $\mathscr{E}$. 

\vspace{2mm}
\noindent
We apply this heuristic algorithm in the following section to obtain two instances of optimal encodings. 

\section{Optimal encoding}
In this section we examine two instances of optimal encodings using reflexive stabilizer codes. The first example is of a channel of qubits inflicted by fully correlated noise. The second example encodes a single qudit on a four state system into four qudits inflicted by single qudit errors. This code is perfect in the same sense as the Perfect Code developed in \cite{perfect} which embeds a single qubit into a five qubit system.

\subsection{Fully correlated noise}\label{fullycorrelated}

We now present our first example to illustrate the power of this novel approach to quantum error correction, and we do so in the case of qubits. This first case we present will be concerned with fully correlated noise, i.e. characterized by the error set $\mathscr{E}=\{\mathbb{1},D_{\overline{1},\overline{0}},D_{\overline{1},\overline{1}}, D_{\overline{0},\overline{1}}\}$. 
These operators 
\begin{equation}\label{fully correlated}
    D_{\overline{1},\overline{0}}=X^{\otimes n},\;\;\;\;
    D_{\overline{1},\overline{1}}=Y^{\otimes n},\;\; \text{and} \;\;
    D_{\overline{0},\overline{1}}=Z^{\otimes n}
\end{equation}
we use the term fully correlated, as whenever a flip, phase, or phase-flip errors occurs, it does so on all qubits simultaneously. In \cite{full} the authors show that a physical system of $n>2$ qubits can protect against fully correlated noise with a maximum number $n-1$ or $n-2$ logical qubits when $n$ is odd or even, respectively, a result which improved on the a similar encoding from \cite{newfull}. With the use of the powerful new tool of reflexive stabilizer codes we show that this previously thought upper-limit, in the even case, of $n-2$ logical qubits can be reduced to $n-1$.
One argument for the physical realization of such noise, is that on a most practical qubit-chips the distance between qubits is often measured in the micrometers, while a likely candidate for environmental noise such as an electromagnetic wave has a wavelength on the order of millimeters. It is this disparity of distances that make it natural to assume that all qubits on the chip are affected by the same error simultaneously. 

As our novel approach shows improvement on a previously thought maximum encoding we take care to give illustrations of the graphs involved for both odd and even $n$ in Figure~\ref{fig: ex2}. To begin by building the avoidance graph for our fully correlated noise, note that the conjugate errors $\mathscr{E}^2=\mathscr{E}$, and hence the edge set of $G_\mathscr{E}$, is independent of $n$. Importantly, the non-looped edges and the loop set for the fully correlated error set are given by 
\begin{equation}\label{fully correlated edges}
    \hat E_\mathscr{E}=\{(\overline{0},\overline{1})\}\;\;\text{and}\;\; 
    \mathscr{L}_{\mathscr{E}}=\{\overline{1}\},
\end{equation} 
respectively. 

\begin{ex}\label{ex2}
We consider any integer $n>2$, and to emphasize the application of the heuristic algorithm, we label the individual steps. 

(S0) Fix the fully correlated error set $\mathscr{E}$ (defined above).
(S1) We can construct the error avoidance graph $G_\mathscr{E}$ by making use of Equation~\eqref{fully correlated edges}. (We illustrate $G_\mathscr{E}$ for $n=3$ qubits at the top left of Figure~\ref{fig: ex2}.) 
(S2) Next, our goal is to find a LUC graph which can avoid the simple edges of $G_\mathscr{E}$. This can be accomplished with the connecting set $C_0=\lan e_1,\ldots, e_{n-1}\ran$. 

Notice that we must extend $C_0$ in order to obtain the maximal encoding of $n-1$ logical qubits. Indeed this is true regardless of the $C_0$ chosen in this step, as dimension arguments would then force $C_1=\{\overline{0}\}$ which violates the condition 
of Theorem~\ref{graph correct}. 
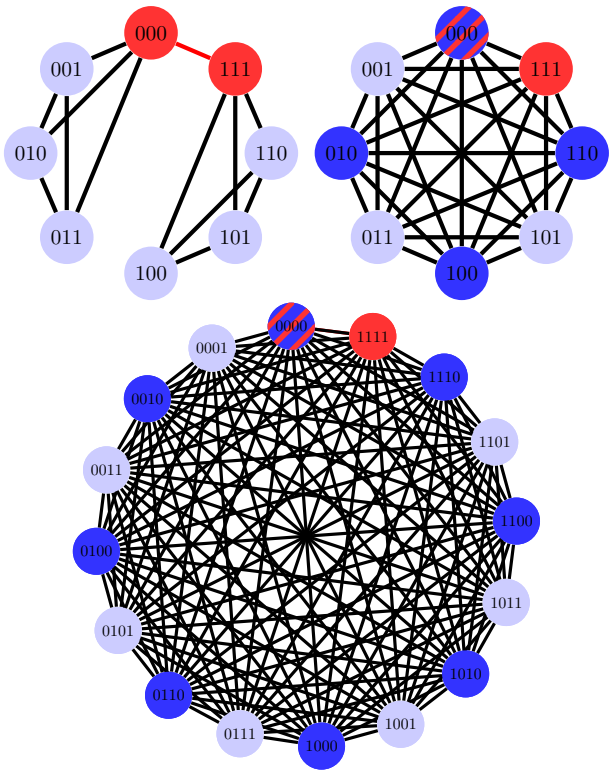
\begin{figure}[h]
\begin{center}

 


 


    


\begin{tikzpicture}
  [scale=.4,auto=left,every node/.style={circle,fill=blue!20, scale=0.8}]
 
 \node[fill=red!80] (n0) at  (0,4)  {000};
 \node (n1) at  (-2.8,2.8) {001};
 \node (n2) at (-4,0) {010};

 \node (n3) at (-2.8,-2.8) {011};

 \node (n4) at (0,-4) {100};
 \node (n5) at (2.8,-2.8) {101};
 
 \node (n6) at (4,0) {110};

\node[fill=red!80] (n7) at (2.8,2.8) {111};

  \foreach \from/\to in {n0/n7}
    \draw[ultra thick,red] (\from) -- (\to);
    
 
  \foreach \from/\to in {n0/n1,n0/n2,n0/n3,n1/n2,n1/n3,n2/n3}
    \draw[ultra thick] (\from) -- (\to);
    
    \foreach \from/\to in {n4/n5,n4/n6,n4/n7,n5/n6,n5/n7,n6/n7}
    \draw[ultra thick] (\from) -- (\to);

\end{tikzpicture}\hspace{1mm}
\begin{tikzpicture}
  [scale=.4,auto=left,every node/.style={circle,fill=blue!20, scale=0.8}]
 
 \node[preaction={fill=blue!80}, pattern=north east hatch, pattern color=red!80, hatch distance=10pt, hatch thickness=2pt] (n0) at  (0,4)  {000};
 \node (n1) at  (-2.8,2.8) {001};
 \node[fill=blue!80] (n2) at (-4,0) {010};

 \node (n3) at (-2.8,-2.8) {011};

 \node[fill=blue!80] (n4) at (0,-4) {100};
 \node (n5) at (2.8,-2.8) {101};
 
 \node[fill=blue!80] (n6) at (4,0) {110};

\node[fill=red!80] (n7) at (2.8,2.8) {111};

  \foreach \from/\to in {n0/n1,n0/n2,n0/n3,n1/n2,n1/n3,n2/n3}
    \draw[ultra thick] (\from) -- (\to);
    
    \foreach \from/\to in {n4/n5,n4/n6,n5/n6}
    \draw[ultra thick] (\from) -- (\to);
    
      \foreach \from/\to in {n0/n4,n0/n5,n0/n6,n1/n4,n1/n5,n1/n6,n2/n4,n2/n5,n2/n6,n3/n4,n3/n5,n3/n6}
    \draw[ultra thick] (\from) -- (\to);
    
    \foreach \from/\to in {n0/n7, n1/n7,n2/n7,n3/n7,n4/n7,n5/n7,n6/n7}
    \draw[ultra thick] (\from) -- (\to);

\end{tikzpicture}

\begin{tikzpicture}
  [scale=.7,auto=left,every node/.style={circle,fill=blue!20, scale=.6}]

 \node (n1) at (-1.79229,3.57599)  {0001};
 \node[fill=blue!80] (n2) at (-3.02,2.612)  {0010};

 \node (n3) at (-3.796,1.261)  {0011};

 \node[fill=blue!80] (n4) at (-3.989,-0.287393)  {0100};
 
 \node (n5) at (-3.57599,-1.79229) {0101};
 
 \node[fill=blue!80] (n6) at (-2.6179,-3.024)   {0110};

\node (n7) at (-1.26126,-3.756)  {0111};

\node[fill=blue!80] (n8) at (0.287,-3.98966)  {1000};

\node (n9) at (1.79229,-3.57599)  {1001};

\node[fill=blue!80] (n10) at (3.02434,-2.6179) {1010};

\node (n11) at (3.79595,-1.26126)  {1011};

\node[fill=blue!80] (n12) at (3.98966,0.287393)  {1100};

\node (n13) at (3.57599,1.79229)  {1101};

\node[fill=blue!80] (n14) at (2.6179,3.02434)  {1110};

\node[fill=red!80] (n15) at (1.26126,3.79595) {1111};

 \node[preaction={fill=blue!80}, pattern=north east hatch, pattern color=red!80, hatch distance=10pt, hatch thickness=2pt] (n0) at (-0.287393,3.98966)   {0000};

  \foreach \from/\to in {n0/n15}
    \draw[very thick,red] (\from) -- (\to);
 
  \foreach \from/\to in {n0/n1,n0/n2,n0/n3,n0/n4,n0/n5,n0/n6,n0/n7,n1/n2,n1/n3,n1/n4,n1/n5,n1/n6,n1/n7,n2/n3,n2/n4,n2/n5,n2/n6,n2/n7,n3/n4,n3/n5,n3/n6,n3/n7,n4/n5,n4/n6,n4/n7,n5/n6,n5/n7,n6/n7}
    \draw[very thick] (\from) -- (\to);
      \foreach \from/\to in 
      {n0/n8,n0/n9,n0/n10,n0/n11,n0/n12,n0/n13,n0/n14,n0/n15,
      n1/n8,n1/n9,n1/n10,n1/n11,n1/n12,n1/n13,n1/n14,n1/n15,
      n2/n8,n2/n9,n2/n10,n2/n11,n2/n12,n2/n13,n2/n14,n2/n15,
      n3/n8,n3/n9,n3/n10,n3/n11,n3/n12,n3/n13,n3/n14,n3/n15,
      n4/n8,n4/n9,n4/n10,n4/n11,n4/n12,n4/n13,n4/n14,n4/n15,
      n5/n8,n5/n9,n5/n10,n5/n11,n5/n12,n5/n13,n5/n14,n5/n15,
      n6/n8,n6/n9,n6/n10,n6/n11,n6/n12,n6/n13,n6/n14,n6/n15,
      n7/n8,n7/n9,n7/n10,n7/n11,n7/n12,n7/n13,n7/n14,n7/n15}
    \draw[very thick] (\from) -- (\to);
    
    \foreach \from/\to in {n8/n9,n8/n10,n8/n11,n8/n12,n8/n13,n8/n14,n8/n15,n9/n10,n9/n11,n9/n12,n9/n13,n9/n14,n9/n15,n10/n11,n10/n12,n10/n13,n10/n14,n10/n15,n11/n12,n11/n13,n11/n14,n11/n15,n12/n13,n12/n14,n12/n15,n13/n14,n13/n15,n14/n15}
    \draw[very thick] (\from) -- (\to);

 \node (n1) at (-1.79229,3.57599)  {0001};
 \node[fill=blue!80] (n2) at (-3.02,2.612)  {0010};

 \node (n3) at (-3.796,1.261)  {0011};

 \node[fill=blue!80] (n4) at (-3.989,-0.287393)  {0100};
 
 \node (n5) at (-3.57599,-1.79229) {0101};
 
 \node[fill=blue!80] (n6) at (-2.6179,-3.024)   {0110};

\node (n7) at (-1.26126,-3.756)  {0111};

\node[fill=blue!80] (n8) at (0.287,-3.98966)  {1000};

\node (n9) at (1.79229,-3.57599)  {1001};

\node[fill=blue!80] (n10) at (3.02434,-2.6179) {1010};

\node (n11) at (3.79595,-1.26126)  {1011};

\node[fill=blue!80] (n12) at (3.98966,0.287393)  {1100};

\node (n13) at (3.57599,1.79229)  {1101};

\node[fill=blue!80] (n14) at (2.6179,3.02434)  {1110};

\node[fill=red!80] (n15) at (1.26126,3.79595) {1111};

 \node[preaction={fill=blue!80}, pattern=north east hatch, pattern color=red!80, hatch distance=10pt, hatch thickness=2pt] (n0) at (-0.287393,3.98966)   {0000};

\end{tikzpicture}
\end{center}
\caption{Illustration of Example~\ref{ex2} when $n=3$.
(Top) The error avoidance graph $G_\mathscr{E}$ for the fully correlated error is shown in red. Overlaid in blue is the LUC graph of $G_{C_0}$ where $C_0=\langle e_1,e_2\rangle$. The simple edges are disjoint as specified by Step 2 of the heuristic algorithm. 
(Bottom) Following Step 3 of the heuristic algorithm, we extend $C_0$ to the set $C=\mathbb{F}_2^3$ and choose the subset $C_1=\langle e_n\rangle$. The resulting LUC graph $G_C^{C_0}$ contains edges between all eight distinct pairs of vertices. Moreover, the only common edge is incident to $1111$ which lies outside of $C_1^\perp$, satisfying the condition of Theorem~\ref{graph correct}. }
\label{fig: ex2}
\end{figure}

(S3) We extend $C_0$ by adding the remaining basis vector; i.e.\ setting $C=\mathbb{F}_2^n$. Simultaneously, we choose $C_1=\lan e_n\ran$. Then $\mathscr{L}_{C_1}=C_1^\perp$ is the set of all strings with the $n^\text{th}$ entry a zero -- shown as blue nodes in Figure~\ref{fig: ex2} for the case $n=3$ -- and is, moreover, disjoint from $\mathscr{L}_\mathscr{E}$ except at $\overline{0}$. Lastly, the single non-zero endpoint of $\hat E_\mathscr{E}$, $\overline{1}$, is not contained in $C_1^\perp$ (shown top right for $n=3$ and bottom for $n=4$ in Figure~\ref{fig: ex2}). Thus the desired properties of Theorem~\ref{graph correct} are satisfied. 

We conclude that the reflexive stabilizer code $R_C^{C_1}$ must correct any error from the fully correlated error set $\mathscr{E}$. Moreover, $R_C^{C_1}$ encodes 
$n-1$ logical qubits into the system of $n$ physical qubits, obtaining the maximum regardless if $n$ is odd or even. 
\end{ex}

The previous example is summarized in the following theorem. 
\begin{thm}\label{thm: fullycorrelated odd}
For any $n>2$, there exists a reflexive stabilizer code $R_C^{C_1}$ encoding $n-1$ logical qubits into $n$ physical qubits that protects against the fully correlated error set. The RSC is constructed with $C=\mathbb{F}_2^n$ and $C_1=\langle e_n\rangle$. 
\end{thm}

In Theorem~\ref{thm: fullycorrelated odd} we provide constructive examples of RSCs that protect against fully correlated error which encode the maximum number of logical qubits as there is no way to encode $n$ logical qubits into $n$ physical qubits. By surpassing the previously thought maximum encoding in \cite{full} and a similar result in \cite{newfull}, we see the true power of this graph theoretic approach. Encoding the errors as edges that need to be avoided a simple answer arises for a once complicated situation. Furthermore, by simple inspection of the error avoidance graph we can create a new code with the same encoding rate by setting $C_1=\lan v\ran$ for any $v\in\mathbb{F}_2^n$ with the only condition  that the weight of $v$ is odd.

\subsection{Perfect code in a 4-state system}\label{perfectcoding}

We now construct a perfect code in a 4-state system analogous to the Perfect Code for qubits given in \cite{perfect}. This reflexive stabilizer code achieves the optimal encoding of a single qudit which protects against single qudit errors.

The basic principal in quantum error correction is the concept that each error transforms distinct code words into distinct orthogonal subspaces. This becomes quite restrictive on the minimal number of physical qudits one can embed into. 
For now, we present only the minimal length for a specific example, namely the case of a 4-state system ($d=4$). This topic, in its full generality, is the subject of future work. 

In analogy to \cite{perfect}, we are interested in protecting against single qudit flip and phase errors. This error set on $n$ qudits is given by 
\begin{equation}
\mathscr{E}=\{\mathbb{1},\;\alpha D_{e_i,\overline{0}},\;\alpha D_{\overline{0},e_i}\;\;|\;\;1\le i\le n,\; 1\le \alpha \le 3\}
\end{equation}
In other words, the embedding space requires an orthogonal subspace for each of the 3 flip and 3 phase errors on each qudit plus one for the unperturbed state. This makes a total of $6n+1$ errors to protect against. To encode $k=1$ logical qudits, we must quadruple this to have enough space to accommodate for each of the $d=4$ embedded states.  Thus, we require $4(6n+1)$ distinct dimension in our Hilbert space. That is, we have the condition
\begin{equation}\label{opt boi}
4(6n+1)\leq 4^n.
\end{equation}
The smallest number satisfying this equation is $n=4$ meaning we must have four physical qudits to encode a single logical qudit. 

Before proceeding, we summarize the properties of the error avoidance graph for single qudit errors. 
\begin{thm}\label{single graph}
Let $\mathscr{E}$ be the set of single qudit errors on a $d$-state quantum system of $n$ qudits. Then $G_\mathscr{E}$ has loops at all vertices with exactly one non-zero entry and $\overline{0}$, and has simple edges between $\overline{0}$ and the vertices of weight two and between distinct vertices of weight one:
\begin{equation}
E_\mathscr{E}=\{(a,b)\;|\; \op{w}(a)=\op{w}(b)=1\;\text{or}\;a=\overline 0\;\text{and}\;\op{w}(b)\le 2\}.
\end{equation}
\end{thm}

By plotting the error avoidance graph for small $n$, one immediately sees the difficulty of avoiding all single qudit errors. As an example, we show $G_\mathscr{E}$ for $n=2$ in Figure~\ref{no work}. 


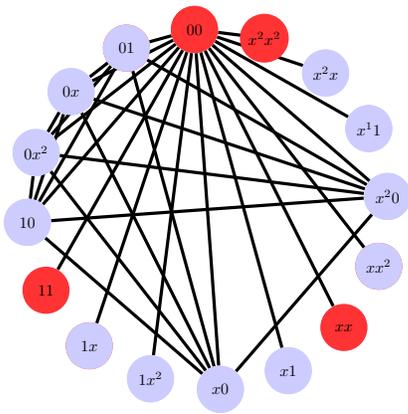
\begin{figure}[h]
\begin{center}
\begin{tikzpicture}
  [scale=.6,auto=left,every node/.style={circle,fill=blue!20, scale=.6}]
 \node[fill=red!80] (n1) at  (-1.79229,3.57599)  {0000};
 \node (n2) at  (-3.02,2.612) {0001};
 \node (n3) at (-3.796,1.261) {0010};
 \node (n4) at (-3.989,-0.287393) {0011};
 \node (n5) at (-3.57599,-1.79229) {0100};
 \node[fill=red!80] (n6) at (-2.6179,-3.024) {0101};
 \node (n7) at (-1.26126,-3.756) {0110};
\node (n8) at (0.287,-3.98966) {0111};
\node (n9) at (1.79229,-3.57599) {1000};
\node (n10) at (3.02434,-2.6179) {1001};
\node[fill=red!80] (n11) at (3.79595,-1.26126) {1010};
\node (n12) at (3.98966,0.287393) {1011};
\node (n13) at (3.57599,1.79229) {1100};
\node (n14) at (2.6179,3.02434) {1101};
\node (n15) at (1.26126,3.79595) {1110};
\node[fill=red!80] (n0) at (-0.287393,3.98966) {1111};
  \foreach \from/\to in {n0/n1,n0/n2,n0/n3,n0/n4,n0/n5,n0/n6,n0/n7,n0/n8,n0/n9,n0/n10,n0/n11,n0/n12,n0/n13,n0/n14,n0/n15}
    \draw[very thick] (\from) -- (\to);
    
      \foreach \from/\to in {n1/n2,n1/n3,n1/n4,n1/n8,n1/n12,n2/n3,n2/n4,n2/n8,n2/n12,n3/n4,n3/n8,n3/n12,
      n4/n8,n4/n12,n8/n12}
    \draw[very thick] (\from) -- (\to);
 
    
    
    
    
    
    
    
    

 \node[fill=red!80] (n0) at  (-0.287393,3.98966)   {0000};
 \node (n1) at (-1.79229,3.57599) {0001};
 \node (n2) at (-3.02,2.612){0010};
 \node (n3) at  (-3.796,1.261)   {0011};
 \node (n4) at (-3.989,-0.287393)  {0100};
 \node[fill=red!80] (n5) at (-3.57599,-1.79229)  {0101};
 \node (n6) at (-2.6179,-3.024)  {0110};
\node (n7) at (-1.26126,-3.756) {0111};
\node (n8) at  (0.287,-3.98966)  {1000};
\node (n9) at (1.79229,-3.57599) {1001};
\node[fill=red!80] (n10) at  (3.02434,-2.6179) {1010};
\node (n11) at  (3.79595,-1.26126) {1011};
\node (n12) at (3.98966,0.287393) {1100};
\node (n13) at  (3.57599,1.79229)  {1101};
\node (n14) at (2.6179,3.02434) {1110};
\node[fill=red!80] (n15) at  (1.26126,3.79595) {1111};
    
\node[fill=red!80] (n0) at  (-0.287393,3.98966)  {00};
\node (n1) at (-1.79229,3.57599)   {01};
\node (n2) at (-3.02,2.612) {0$x$};
\node (n3) at  (-3.796,1.261) {0$x^2$};
\node (n4) at  (-3.989,-0.287393) {10};
\node[fill=red!80] (n5) at  (-3.57599,-1.79229)  {11};
\node (n6) at (-2.6179,-3.024) {1$x$};
\node (n7) at  (-1.26126,-3.756) {1$x^2$};
\node (n8) at (0.287,-3.98966)   {$x$0};
\node (n9) at (1.79229,-3.57599)  {$x$1};
\node[fill=red!80] (n10) at (3.02434,-2.6179) {$xx$};
\node (n11) at  (3.79595,-1.26126) {$xx^2$};
\node (n12) at  (3.98966,0.287393) {$x^2$0};
\node (n13) at  (3.57599,1.79229) {$x^1$1};
\node (n14) at  (2.6179,3.02434) {$x^2x$};
\node[fill=red!80] (n15) at  (1.26126,3.79595) {$x^2x^2$};

\end{tikzpicture}

\end{center}
\caption{The error avoidance graph $G_\mathscr{E}$ for the set of all single qudit errors on a four-state quantum system of two qudits. The field of four elements is denoted as $\mathbb{F}_4=\{0,1,x,x^2\}$ where $x^2+x+1=0$. The quantity of edges makes it impossible to find an RSC to correct against all possible errors. (See also Equation~\ref{opt boi}.)  }
\label{no work}
\end{figure}




As a consequence of Theorem~\ref{single graph}, we have the following corollary. The details are expounded in Appendix~\ref{app: benchmark}.
\begin{cor}\label{single correct}
Let $C_1\subset C\subset\mathbb{F}_d^n$ be linear subspaces such that $\op{wt}(C)\geq 3$ and $\op{wt}(C_1^\perp)\geq 2$. The reflexive stabilizer code $R_C^{C_1}$ can correct any single qudit error. 
\end{cor}

We are now ready to construct our perfect code on a four-state system. 
\begin{ex}\label{4statex}
Let $\mathscr{E}$ be the set of single qudit errors described above. 
The perfect code for a four-state system will embed a single logical qudit into the optimal four physical qudit system, set by Equation~\ref{opt boi}. 
First, set $C\subset\mathbb{F}_4^4$ to be the 2-dimensional connection set consisting of the following vectors:
\begin{center}
\begin{tabular}{cccc}
$0000$,&$1x10$,&$xx^2x0$,&$x^21x^20$,\\
$x^2x^201$,&$110x$,&$xx0x^2$,&$x111$\\
$x^2xxx$,&$1x^2x^2x^2$,&$10x1$,&$x0x^2x$,\\
$x^201x^2$,&$0x^21x$,&$01xx^2$,&$0xx^21$.
\end{tabular}
\end{center}
Note that $\op{wt}(C)=3$. According to Corollary~\ref{single correct}, we must find a $C_1\subset C$ of dimension 1 such that $\op{wt}(C_1^\perp)\ge 2$. 
The subset $C_1=\{0000,x111,x^2xxx,1x^2x^2x^2\}$ satisfies these conditions. Thus $R_C^{C_1}$ can correct any single qudit error. 
\end{ex}


\comment{
\begin{figure}[h]
\begin{center}
\begin{tikzpicture}[scale=0.6 ,transform shape]
  \foreach \number in {1,...,8}{
        \mycount=\number
        \advance\mycount by -1
  \multiply\mycount by 45
        \advance\mycount by 0
      \node[draw,circle,inner sep=0.25cm] (N-\number) at (\the\mycount:5.4cm) {};
    }
  \foreach \number in {9,...,16}{
        \mycount=\number
        \advance\mycount by -1
  \multiply\mycount by 45
        \advance\mycount by 22.5
      \node[draw,circle,inner sep=0.25cm] (N-\number) at (\the\mycount:5.4cm) {};
    }
  \foreach \number in {1,...,15}{
        \mycount=\number
        \advance\mycount by 1
  \foreach \numbera in {\the\mycount,...,16}{
    \path (N-\number) edge[->,bend right=3] (N-\numbera)  edge[<-,bend
      left=3] (N-\numbera);
  }
}
\end{tikzpicture}
\end{center}
\caption{\DW{Add references and words.}}
\label{no work}
\end{figure}
}

\section{Discussion}\label{discussion}
In this work, we introduced a novel approach to quantum error correction motivated by graph theory. We developed two graphs -- error avoidance graphs and LUC graphs -- to visualize an error sets and reflexive stabilizer codes, respectively, and repose the algebraic conditions of error correction in terms of edge avoidance. This approach, summarized as a heuristic algorithm, places the error sets at the forefront by providing a means to construct an encoding that protects against a predetermined noisy channel. This viewpoint is attractive because it allows for the development of codes that protect errors in a variety quantum computer architectures regardless of the intrinsic set of errors present. Furthermore, this viewpoint promotes a collaborative mindset by recognizing that the engineers tasked with developing a quantum computer have limited control over the suppression of errors. 

Another benefit to this approach is that it allows for the correction of correlated errors directly and without additional assumptions. For instance, it is common in the literature to presume independence of errors. When applied to a set of correlated errors, this presumption manifests by effectively requiring for correction against a larger error set which can lead to lower rates of error correction and fidelity \cite{temporalcorrelation}. 
Moreover, recent experimental observations of correlated errors bring into question the validity of this assumption \cite{correlated}.  For these reasons, the ability to correct correlated errors has become increasingly relevant. 

Lastly, we argue that RSCs are not only easy to use, but practical as well: 
We have benchmarked the error correction rates for RSCs against the industry-standard CSS codes, showing that RSCs have the same capabilities for single qubit error correction. We have developed RSCs in a framework that allows for error correction on multi-state quantum system represented as qudits, generalizing the two-state system represented by qubits. The relevance of this framework is supported by recent experiments demonstrating that more than two energy levels are measurable in a system of silicon-based quantum dots \cite{quditsnotqubits}. In addition, we present two constructive instances of optimal encodings: a maximal encoding of qubits that corrects fully correlated noise, and a perfect code which minimally encodes a single qudit on a four-state system against single qudit errors. The former example improves on the previously-proven ``optimal" encoding rate shown in \cite{full} demonstrating even further the ease of use of the graph-theoretic representation. 


This introductory work on reflexive stabilizer codes establishes the utility of the graph-theoretic approach employed, and displays a number of practical and theoretical applications. However, with the new approach comes a wealth of unanswered questions and avenues for future research. The connection between quantum error correcting codes and edge avoidance in graphs has opened up a rich vein of future research opportunities. Below we provide a small list of questions to help guide future explorations. 

The first set of questions address fundamental existence and uniqueness conditions for error sets and reflexive stabilizer codes. Their wording is designed to make them accessible to researchers in graph theory and quantum information science, and to promote collaboration between the same. 


\begin{quest}{1}
What necessary and sufficient conditions on a given error set $\mathscr{E}$ or, equivalently, an error avoidance graph $G_{\mathscr{E}}$ guarantee the existence of a RSC $R_C^{C_1}$ or, equivalently, a LUC graph $G_C^{C_1}$ that corrects those errors?
\end{quest}

\begin{quest}{2}
Given two error sets $\mathscr{E}$ and $\mathscr{E}'$ or, equivalently, two error avoidance graphs $G_{\mathscr{E}}$ and $G_{\mathscr{E}'}$, what properties guarantee a common RSC to protect against each set? 
\end{quest}

\begin{quest}{3}
When is it true that two given error sets $\mathscr{E}$ and $\mathscr{E}'$ produce isomorphic error avoidance graphs $G_\mathscr{E}$ and $G_{\mathscr{E}'}$. 
\end{quest}

The last two questions are more specific in scope. First, we look at quantum random walks. This field of study already lies at the intersection of quantum information and algebraic graph theory. Moreover, they have been shown to be universal for quantum computation by exploiting perfect or group state transfer on graphs\cite{pst,fracrevival,gst}. 
The last question is an option to incorporate graph theory techniques into the study of quantum error correcting codes, and is related to the works \cite{graphons,gvbound}. 

\begin{quest}{4}
What LUC graphs have state transfer with quantum random walks, continuous or discrete? \cite{cayleywalk} 
\end{quest}


\begin{quest}{5}
Using limiting properties of graphs or graphons, can one find a GV-Bound for reflexive stabilizer codes? 
\end{quest}

\begin{acknowledgements}
We would like to thank Alastair Kay and David Feder for their useful conversation and insight. We would also like to thank the referees for their comments which helped reshape this manuscript. 
\end{acknowledgements}

\appendix



\section{Constructing Reflexive Stabilizer Code}\label{app: constructive}
Here we provide a constructive form for reflexive stabilizer codes. To do so, we must construct a joint eigenspace for the stabilizer $S_C^{C_1}$. 

First, we examine the eigenspaces of $X(1)Z(1)$ in $\mathbb{C}^p$, where $p$ is a prime. The eigenvalues for $X(1)Z(1)$ are the $p^\text{th}$ roots of unity $\omega^\kappa$, where each eigenspace is one-dimensional. 
The eigenvalue $\omega^\kappa$
is spanned by the eigenstate
\begin{equation}\label{fp states}
|\psi_\kappa\ran=\frac{1}{\sqrt{d}}\sum_{a\in\mathbb{F}_p}\alpha_{a}|a\ran 
\end{equation}
where $\alpha_0=1$, $\alpha_{p-1}=\omega^{\kappa+1}$ and, for $1\leq i\leq p-2$, $\alpha_i=\omega^{T^\kappa_i}$ for $T^\kappa_i=\frac{i(i-1-2\kappa)}{2}$.

Notice that, for qubits ($d=2$), the eigenstates for $\pm \iota$ are exactly the conjugate (Hadamard) basis states 
\begin{align}\label{p2}
|\psi_0\ran&=\frac{1}{\sqrt{2}}\left(\iota|0\ran+|1\ran\right)\\
|\psi_1\ran&=\frac{1}{\sqrt{2}}\left(-\iota|0\ran+|1\ran\right).
\end{align}
For all other $a\in\mathbb{F}_p$, Equation~\eqref{eqcomm} yields the relation 
\begin{equation}
X(a)Z(a)=\omega^{\frac{a(a-1)}{2}}\left(X(1)Z(1)\right)^{a}    
\end{equation}
Hence the eigenstates for $X(a)Z(a)$ are exactly those for $X(1)Z(1)$ given in Equation~\eqref{fp states}.
We can then extend to $a=\sum_{i=1}^m\alpha_if_i\in\mathbb{F}_d$, using the alternative definition for the Pauli operators
given in \cite{nonbinary}, to get 
\begin{align}\label{prf}
X(a)Z(a)&=\omega^{\tau_a}\bigotimes_{i=1}^m\left(X(1)Z(1)\right)^{\alpha_i}\\
\tau_a&=\displaystyle\frac{1}{2}\sum_{i=1}^m\alpha_i(\alpha_i-1).
\end{align}
Equation~\eqref{prf} is used to extend the states $|\Psi_\kappa\ran$ in Equation~\eqref{fp states} to eigenstates for each $a\in\mathbb{F}_d$. By taking tensor products, we extend further to $ a\in\mathbb{F}_d^n$ for an arbitrary $n$-state, $d$-level quantum system.

Given a LUC $G(C)$ and a linear subspace $C_1\subset C$ we define the \textbf{reflexive quantum stabilizer code} (RSC) of $G(C)$ and $C_1$ as
\begin{equation}
R^C_{C_1}=\left.\left\{|\Phi_{c'}\ran=\frac{1}{\sqrt{|C_1|}}\sum_{c\in C_1}D_{c,\overline{0}}|\Psi_{c'}\ran\,\right|\,c'\in C\right\},
\end{equation}
where $|\Psi\ran=\bigotimes_{i=1}^n |\psi_{c_i}\ran$, for $c=(c_1,...,c_n)$.

\section{Lemmas for Error Correcting Theorem}\label{app: lemmas}
This section contains the lemmas leading up to Theorem~\ref{graph correct}. We will restate the conditions of Theorem~\ref{errorcorrecting}, given again below, in terms of the edges of the LUC and error avoidance graphs. 
The first lemma gives conditions for a conjugate error to avoid $\big(S_C^{C_1}\big)^\perp$ of a reflexive stabilizer code $R_C^{C_1}$. 

\begin{lem}\label{graph correct lemma 2}
Let $d=p^m$ for some prime $p$, $C_1\subset C\subset\mathbb{F}_d^n$ be linear subspaces, and $\mathscr{E}$ an error set. The set of non-trivial conjugate errors $E=\omega^\kappa D_{a,b}\in\mathscr{E}^2$ that lie outside of $\big(S_C^{C_1}\big)^\perp$; i.e.\ the set $\mathscr{E}^2\setminus\big(\big(S_C^{C_1}\big)^\perp\cup \mathbb{1}\big)$, is characterized by the following graph-theoretic relation on $G_\mathscr{E}$ and $G_C^{C_1}$: 
\begin{equation}\label{graph correct eqn 5}
E_C^{C_1}\cap {E}_\mathscr{E}\subseteq\{(a,b)\,|\,a\nin C_1^\perp\}.
\end{equation}
\end{lem}
\begin{proof}
Recall the form of $\big(S_C^{C_1}\big)^\perp$ given in Lemma~\eqref{sperp}: 
\begin{equation}
\left(S_C^{C_1}\right)^\perp=\left\langle D_{ab}\;\Big|\; a-b\in C,\;\text{and } a\in C_1^\perp \right\rangle.
\end{equation}
Fix a conjugate error $E=\omega^\kappa D_{a,b}\in\mathscr{E}^2$. 
First, suppose that $a\neq b$ and therefore $(a,b)\in\hat E_\mathscr{E}$. 
Then, by Lemma~\ref{sperp}, $E\nin \big(S_C^{C_1}\big)^\perp$ exactly when $b-a\nin C$; i.e.\ $(a,b)\nin E_C^{C_1}$, or $a\nin C_1^\perp$. Equation~\eqref{graph correct eqn 5} is exactly this condition when restricted to the simple edges.  

Next, suppose that $a=b$. Then, since $\overline 0\in C$, $E\nin \big(S_C^{C_1}\big)^\perp$ exactly when $a\nin C_1^\perp$. Equation~\eqref{graph correct eqn 5} is exactly this condition when applied to loops. 
\end{proof}

The next lemma gives conditions for a conjugate error to be in $S_C^{C_1}$ of a reflexive stabilizer code $R_C^{C_1}$. 

\begin{lem}\label{graph correct lemma 1}
Let $d=p^m$ for some prime $p$, $C_1\subset C\subset\mathbb{F}_d^n$ be linear subspaces, and $\mathscr{E}$ an error set. The set of conjugate errors $E=\omega^\kappa D_{a,b}\in\mathscr{E}^2$ that lie inside of $S_C^{C_1}$; i.e.\ the set $\mathscr{E}^2\cap S_C^{C_1}$, is characterized by the following graph-theoretic relations on $G_\mathscr{E}$ and $G_C^{C_1}$: 
\begin{equation}\label{graph correct eqn 3}
\hat E_C^{C_1}\cap \hat{E}_\mathscr{E}\subseteq\{(a,b)\,|\,a-b\in C_1\}
\end{equation}
and 
\begin{equation}\label{graph correct eqn 4}
\mathscr{L}_\mathscr{E}\subseteq C^\perp. 
\end{equation}
\end{lem}
\begin{proof}
Fix a conjugate error $E=\omega^\kappa D_{a,b}\in\mathscr{E}^2$. 
First, suppose that $a\neq b$ and therefore $(a,b)\in\hat E_\mathscr{E}$. 
Then, by Equation~\eqref{refl stab}, $E\in S_C^{C_1}$ exactly when $a-b\in C_1$. Equation~\eqref{graph correct eqn 3} is exactly this condition. 

Next, suppose that $a=b$. Then, by Equation~\eqref{refl stab}, $E\in S_C^{C_1}$ exactly when $a\in C^\perp$. Equation~\eqref{graph correct eqn 4} is exactly this condition. 
\end{proof}

Taking Lemma~\ref{graph correct lemma 2} and the fact that $\overline 0\in C_1$ from Lemma~\ref{graph correct lemma 2} yields Theorem~\ref{graph correct}. 
Below is the most general version of Theorem~\ref{graph correct}.

\begin{thm}\label{graph correct thm 2}
Let $d=p^m$ for some prime $p$, $C_1\subset C\subset\mathbb{F}_d^n$ be linear subspaces, and $G_C^{C_1}$ be a LUC graph. Then $R^C_{C_1}$ can correct any error from an error set $\mathscr{E}$, if 
\begin{equation}\label{graph correct eqn 1}
\hat E_C^{C_1}\cap \hat{E}_\mathscr{E}\subseteq\{(a,b)\,|\,a\nin C_1^\perp\;\text{or}\;a-b\in C_1\}
\end{equation}
and 
\begin{equation}\label{graph correct eqn 2}
\mathscr{L}_\mathscr{E}\cap C^\perp_1\subseteq C^\perp. 
\end{equation}
\end{thm}


\section{Comparison to CSS Codes}\label{app: benchmark}

 Reflexive stabilizer codes and CSS codes are in one-to-one correspondence via a change in error basis. That is, choosing a basis for the error group  which is instead generated by $Y$ and $Z$ rather than $X$ and $Y$ from Example~\ref{qubit basis}. Specifically, if we define:
\begin{equation}
\tilde{D}_{\overline{a},\overline{b}}=\bigotimes_{i=1}^n Y^{a_i}Z^{b_i},
\end{equation}
a natural isomorphism arises between the two codes. By maintaining the standard basis on $\mathbb{F}_2^{2n}$, i.e. the parameter space of the Error Basis, we can build a linear isomorphism
\begin{align}\label{boom}
\varphi:\mathbb{F}_2^{2n}&\rightarrow\mathbb{F}_2^{2n}\\
(a,a)&\mapsto (a,0)\\
(b,0)&\mapsto (0,b),
\end{align}
which induces an automorhpism on $\mathcal{E}_n$, $\Phi:\mathcal{E}_n\rightarrow \mathcal{E}_n$. It is now a simple exercise to show that the above isomorphism takes the stabilizer for a reflexive stabilizer code $R^C_{C_1}$ (with $d=2$) to a stabilizer of a CSS code. Further, one easily verifies that $\lan a,b\ran=0$ if and only if $\lan \varphi(a),\varphi(b)\ran=0$, and thus $R_C^{C_1}$, with stabilizer $S$, can correct any error from $\mathscr{E}$ if and only if $\Phi(S)$ induces a CSS code which corrects any error from $\Phi(\mathscr{E})$ for any error set $\mathscr{E}$.

The traditional approaches of error correction for CSS codes is that of considering $t$ single qubit flip, phase, or phase-flip errors. That is, we consider the error set:
\begin{equation}
\mathscr{E}=\{D_{p,0},D_{0,p},D_{p,p}\}
\end{equation}
where $p=\sum a_ie_i$ where no more than $t$ of the $a_i=1$ and rest are zero, i.e. the error set of at most $t$-flips, $t$-phases, and $t$-phase-flips. For this error set we note the following relationship
\begin{equation}
\Phi(\mathscr{E})=\mathscr{E},
\end{equation}
and hence the reflexive code obtained as an image of a CSS code under the automorphism $\Phi\inv$ corrects the same $t$ single qubit errors. Therefore the next theorem follows directly from the isomorphism in Equation~\ref{boom} and Theorem 1 in \cite{goodcodes}.
\begin{thm}\label{error ext}
Let $C_2\subset C_1\subset\mathbb{F}_d^n$ be linear subspaces, then the code $R_{C_1}^{C_2}$ can both correct up to $t$ single-qubit flip, phase, or phase-flip errors where
\begin{equation}
t=\op{min}\left\{\left\lfloor\frac{\op{wt}(C_1)-1}{2}\right\rfloor,\left\lfloor\frac{\op{wt}\left(C_2^\perp\setminus C_1\right)-1}{2}\right\rfloor\right\}.
\end{equation}
\end{thm}
Additionally, by the same relationship we arrive at the following result from Theorem 2 in \cite{ortho}. 
\begin{thm}
There exists a family of reflexive stabilizer codes with asymptotic rate 
\begin{equation}
R=1-2\delta\log_2(3)-H_2(2\delta)
\end{equation}
where $\delta$ is the fraction of qubits that are subject to decoherence and $H_2(\delta)=-\delta\log_2(\delta)-(1-\delta)\log_2(1-\delta)$ is the binary entropy function.
\end{thm}

\bibliography{references}

\end{document}